\documentclass[11pt]{article}
\usepackage{amsmath,amsthm,amssymb}
\usepackage{mathrsfs}
\usepackage{array, longtable}
\usepackage{pdfsync}
\usepackage{xcolor}
\newtheorem{Theorem}{Theorem}[section]
\newtheorem{lem}[Theorem]{Lemma}
\newtheorem{Remark}[Theorem]{Remark}
\newtheorem{Definition}[Theorem]{Definition}
\newtheorem{Corollary}[Theorem]{Corollary}

\newtheorem{Example}[Theorem]{Example}

\numberwithin{equation}{section}
\newcommand{\dst}[2]{\genfrac{[}{]}{0pt}{0}{#1}{#2}}
\begin{document}

\title {\bf$t$-Fold $s$-Blocking  Sets and $s$-Minimal Codes
}

\author{Hao Chen$^1$ Xu Pan$^2$ and Conghui Xie$^3$
}
\date{\small
$^1$College of Information Science and Technology,
Jinan University, Guangzhou, Guangdong, 510632, China, haochen@jnu.edu.cn\\
$^2$School of Mathematics and Statistics, Central China Normal University, Wuhan, Hubei, 430079, China, panxucode@163.com\\
$^3$Hetao Institute of Mathematics and Interdisciplinary Sciences,  Shenzhen, Guangdong, 518033, China,
  xiech@himis-sz.cn
}

\maketitle

\begin{abstract}
Blocking sets and minimal codes have been studied for many years in projective geometry and coding theory. In this paper, we provide a new lower bound on the size of $t$-fold $s$-blocking sets without the condition $t \leq q$, which is stronger than the classical result of Beutelspacher in 1983. Then a lower bound on lengths of projective $s$-minimal codes is also obtained. It is proved that $(s+1)$-minimal codes are certainly $s$-minimal codes. We generalize the Ashikhmin-Barg condition for minimal
codes to  $s$-minimal codes. Many infinite families of $s$-minimal codes satisfying and violating this generalized Ashikhmin-Barg condition are constructed. We also give several examples which are binary minimal codes, but not $2$-minimal codes.\\


\end{abstract}

\section{Introduction}
Throughout this paper, let $\mathbb{F}_q$ be the finite field of order $q$, where $q$ is a power of a prime. For a positive integer $n$, let $\mathbb{F}_q^{n}$ be the $n$-dimensional vector space over $\mathbb{F}_q$ consisting of all the $n$-tuples $(x_{1},x_{2},\cdots,x_{n}),$ where $x_i\in \mathbb{F}_q$ for all $1\leq i\leq n.$
For a positive integer $n$, assume $[n]$ is the set of positive integers such that $[n]=\{1,2,3,\cdots,n\}$.

A nonempty subset $X$ of $\mathbb{F}_q^{n}$ is called an $(n,|X|)_q$-{\it code}, where $|X|$ is the size of $X$.
A subspace $C$ of $\mathbb{F}_q^{n}$ with dimension $k$ is called an {\it $[n,k]_q$-linear code} or an {\it $[n,k,d]_q$-linear code}, where $d$ is the minimum Hamming distance of $C$. The minimum Hamming distance of a code is the minimum Hamming distance between any two distinct codewords in this code.
The dual code of a linear code $C$ is defined as
$$
C^{\perp}=\{{\bf x}\in \mathbb{F}_q^n \,| \,  \langle{{\bf x},\bf c}\rangle =0 \text{ for any } {\bf c}\in C\},$$
where $\langle{{\bf x},\bf c}\rangle$ is the Euclidean inner product between ${\bf x}$ and ${\bf c}.$
An linear code $C$ is called  {\it projective}, if
the minimum Hamming distance $d(C^\bot)$ of the dual code $C^\bot$ satisfies $d(C^\bot)\ge 3.$


A nonzero codeword ${\bf c}$ of a  $[n,k]_q$-linear code $C$ is called minimal,
if for any codeword ${\bf c}'$ satisfying $\mathrm{supp}({\bf c}') \subseteq \mathrm{supp}({\bf c})$,
then there is a nonzero $\lambda \in \mathbb{F}_q$, such that, ${\bf c}'=\lambda {\bf c}$.
A linear code $C$ is called minimal if each nonzero codeword is minimal.
The Ashikhmin-Barg criterion $\frac{w_{min}}{w_{max}} > \frac{q-1}{q}$ was proposed in \cite{AB} as a sufficient condition
for a $q$-ary linear code to be minimal, where $w_{min}$ and $w_{max}$ are minimum and maximum weights of the code. In \cite{AB}, minimal codewords and minimal codes were introduced for general $q$-ary linear codes and studied systematically.
Minimal codewords in linear codes were originally used
in  decoding algorithms \cite{H79}. Minimal binary linear codes were studied in \cite{CL} in the name linear interesting codes, also see \cite{Borello}.
 Massey \cite{MA} showed that minimal codewords in a linear code determine the access structure in his
code-based secret sharing scheme.  Brassard, Cr\'{e}peau and Santha used binary minimal linear  codes for constructing oblivious transfers of vectors, see \cite{Brassard}.  There are many constructions of minimal codes violating the Ashikhmin-Barg criterion, see \cite{DHZ,Chen} and reference therein. We refer to \cite{CMP,CMR} for variations of minimal codes.

Blocking sets in projective spaces were have been studied for many years, see \cite{Beu,HN}.
A subset $S$ in the $k-1$ dimension projective space $\mathrm{PG}(k-1,q)$ over $\mathbb{F}_q$ is called a cutting $s$-blocking set,
if for each projective subspace $U \subset\mathrm{PG}(k-1,q)$ of codimension $s$, the intersection $S \bigcap U$ span the space $U$. When $s=1$, a cutting $1$-blocking set is also called strong blocking set, see \cite{HN,TQ,AB20,AB22}.
The following geometric characterization of minimal codes was given in \cite{AB20,TQ,HN}.
Let $G=[G^T_{1},G^T_{2},\cdots,G^T_{n}]$ be the generator matrix of an  $[n,k]_q$-linear code $C$,
where $G_1, G_2, \cdots, G_n$ are nonzero vectors in $\mathbb{F}_q^k$.
Let $B$ be the subset of $n$ points in $\mathrm{PG}(k-1,q)$ corresponding to $G_1, G_2, \cdots, G_n$.
Then $C$ is a minimal code if and only if $B$ is a strong blocking set. In \cite{Alfarano2}, this geometric characterization of minimal codes was generalized to minimal codes in the rank matrix.

From the classical results by Beutelspacher in \cite{Beu}, a lower bound on the size of cutting $t$-blocking set can be given, in which a condition on $t$ depending $q$ had to be imposed, see \cite[Theorem 4]{Beu} or \cite[Theorem 1.12]{AB22}.
Then by the using of combinatorial Nullstellensatz, a lower bound on the size of strong cutting set was given in \cite[Theorem 2.14]{AB22}. Let $m(k,q)$ be the length $n$ such that there exists a minimal $[n, k]_q$-linear code.
Their lower bound is $$m(k,q)\ge (k-1)(q+1).$$
An upper bound for $m(k,q)$
$$m(k,q)\leq (q+1)   \Big{\lceil}  \frac{2}{1+\frac{1}{(q+1)^2 \ln q}}(k-1) \Big{\rceil},$$ was proved by a probabilistic argument, see \cite{HN}.
Some new upper bounds on the smallest size of affine blocking sets were proved in \cite{Bi}.



In 1977, the notion of the generalized Hamming weights was introduced  by Helleseth, Kl{\o}ve and Mykkeltveit in \cite{HT}. Wei \cite{W} provided an application of the generalized Hamming weights in wire-tap channels of type II.
Since then, lots of works have been done in computing and describing the generalized Hamming weights for many classes of linear codes, see \cite{TV,HK95,HP1}. In a recent paper \cite{GS}, subcodes with minimal supports were used to consider pure resolutions and Betti numbers of linear codes. A $s$-dimension subcode $D \subset C$ is called minimal, if there is a $s$-dimension subcode $D'$ such that $\mathrm{supp}(D') \subset \mathrm{supp}(D)$, then $D'=D$. It is clear that this is a generalization of minimal codewords in the case $s=1$. A linear code is called $s$-minimal, if each of its $s$-dimension subcode is minimal. From the view of coding theory, it is natural to study $s$-minimal codes and cutting $s$-block sets. $s$-minimal codes have been studied in \cite{XKH}.



Alon, Bishnoi, Das and Neri constructed strong blocking sets of size at most $Cqk$ in the projective space $PG(k-1,q)$ in their celebrated paper \cite{Alon}, where $C$ is an absolute constant. This is equivalent to the explicit constructions of minimal codes of length at most $Cqk$. Their result was generalized to an explicit construction of strong $s$-blocking sets (cutting $s$-blocking sets in this paper) of size at most $2^{O(s^2logs)}q^sk$ in \cite{Bis}. This is also equivalent to the explicit construction of $s$-minimal codes of length at most $2^{O(s^2logs)}q^sk$.


 In this paper, we prove a new lower bound on the sizes of $t$-fold $s$-blocking sets (see Theorem 3.1), which is stronger than the classical result in \cite{Beu}. Let $m(k,s,q)$ be the minimal length $n$ such that there exists a $s$-minimal $[n, k]_q$-linear code.
 Then we obtain a lower  bound on $m(k,s,q)$ without condition $s \leq q$. Moreover a characterization of $s$-minimal codes using cutting $s$-blocking set is obtained. It is also proved that $(s+1)$-minimal codes are $s$-minimal codes, see Theorem 4.5. Several examples of minimal codes which is not $2$-minimal are constructed, see Example 4.7 and 4.8.
 A generalized Ashikhminand-Barg condition which is sufficient for $s$-minimal
 code is proved. Many $s$-minimal codes satisfying the generalized Ashkhmin-Barg condition or violating the Ashikhmin-Barg condition are also constructed, see Section 7 and 8.

The rest of the paper is organized as follows: Section 2 gives some preliminaries and some notations.  Section 3 provides some new lower bound for the size of $t$-fold $s$-blocking sets.
In Section 4, we obtain a characterization of s-minimal codes as cutting $s$-blocking
sets.
In Section 5,  we study some properties of $s$-minimal subcodes.
 In Section 6,  we  generalize of the Ashikhminand-Barg theorem for $s$-minimal codes.
In Section 7,  we construct several infinite families of $s$-minimal codes from Solomon-Stiffler codes.
In section 8, we give constructions of $s$-minimal codes that violate the generalized Ashikhmin-Barg condition.

\section{Preliminaries}
This section introduces some notions and basic properties used in this paper.
Let $\mathrm{PG}(k-1,q)$ be the
finite projective geometry of dimension $k-1$ over  $\mathbb{F}_q^{k}$.
The $t$-fold $s$-blocking set in $\mathrm{PG}(k-1,q)$ is defined in the following definition.

\begin{Definition}
 Let $t$, $s$ and $k$ be positive integers such that $1\leq s\leq k-1.$
  A $t$-fold $s$-blocking set of $\mathrm{PG}(k-1,q)$ is a set $B\subseteq \mathrm{PG}(k-1,q)$ such that $| B \bigcap U|\ge t$ for every projective subspace $U \subset\mathrm{PG}(k-1,q)$ of codimension $s$.
 When $s = t = 1$, $B$ is simply called a blocking set.
\end{Definition}

A cutting $s$-blocking set of $\mathrm{PG}(k-1,q)$  was introduced in  \cite{BB}.
A cutting $s$-blocking set is actually a $(k-s)$-fold strong blocking set, which was introduced
independently in \cite{DG11}.

\begin{Definition}[\cite{BB}]
 Let $s$ and $k$ be positive integers such that $1\leq s\leq k-1.$
 A cutting $s$-blocking set of $\mathrm{PG}(k-1,q)$ is a  set $B\subseteq\mathrm{PG}(k-1,q)$ such that $U$ is generated by $B \bigcap U$ for every projective subspace $U \subset\mathrm{PG}(k-1,q)$ of codimension $s$.
 A cutting $1$-blocking set is called a cutting blocking set (also strong blocking set).
\end{Definition}


For a vector space $U$ over $\mathbb{F}_q$ of dimension $k$, we assume that
 $$\,\,{\rm SUB}^{s}(U)=\{V\,|\,V\text{ is a subspace of $U$ with dimension }s\}$$ for $0\leq s\leq k.$
For a  nonempty subset $X$ of $\mathbb{F}_q^n$, the {\it support} of $X$ is $$\mathrm{supp}(X)=\{j\in [n] \,|\,\text{there exists }\,\mathbf{x}=(x_{1}, x_2, \cdots,x_{n})\in X\text{ with } x_{j}\neq0\},$$
and the {\it support weight} of $X$ is $w(X)=  |\mathrm{supp}(X)|.$
For any vector $\mathbf{x}\in \mathbb{F}_q^{n}$,   the  support weight of the subspace generated by $\mathbf{x}$ is equal to the Hamming weight of $\mathbf{x}$ denoted by $w(\mathbf{x}).$

\begin{Definition}
For an $[n,k]_q$-linear code $C$ and $1\leq s\leq k$, the  {\it$s$-generalized Hamming weight} ($s$-GHW) of $C$ is defined as $$d_{s}(C)=\min\{|\mathrm{supp}(U)|\,\big|\,  U \in {\rm SUB}^{s}(C)\}.$$
The set $\{d_{1}(C),d_{2}(C),\cdots,d_{k}(C)\}$ is called the  {\it generalized Hamming weight hierarchy} of $C$.
\end{Definition}

When $s=1$, the parameter $d_{1}(C)$ is the minimum Hamming weight of $C$.
For an $[n,k]_q$-linear code $C$, the {\it monotonicity bound on the generalized Hamming weights} \cite{W} is provided as following:
$$1\leq d_1(C) < d_2(C)<\cdots < d_k(C) \leq  n.$$
 The monotonicity bound on the generalized Hamming weights produces {\it the Singleton bound on the generalized Hamming weights} as following: $$d_s(C) \leq  n-k+s $$ for $1\leq s\leq k.$
The following Griesmer bounds for generalized Hamming weight of linear codes are proved in Theorem 1 of  \cite{HK95}.
  \begin{lem} \label{L2}
Let $C$ be an $[n,k]_q$-linear code with the generalized Hamming weight hierarchy $\{d_{1}(C),d_{2}(C),\cdots,d_{k}(C)\}$.
\begin{description}
  \item[(a)] For $1\leq s\leq r\leq k,$
$d_{s}(C)$ and $d_{r}(C)$ satisfy $$ \frac{q^r -1}{q^r-q^{r-s}}d_{s}(C) \leq  d_{r}(C)  .$$
In particular,  when $r=s+1$, $d_{s}(C)$ satisfies $ \frac{q^{s+1} -1}{q^{s+1}-q}d_{s}(C) \leq  d_{s+1}(C)  .$
  \item[(b)] For $1\leq s\leq k,$ $d_{s}(C)$ satisfies $d_s(C)+\sum_{i=1}^{k-s}\lceil \frac{(q-1)d_s(C)}{q^i(q^s-1)}\rceil \leq  n.$
\end{description}
 \end{lem}

\begin{Definition}
For an $[n,k]_q$-linear code $C$ and $1\leq s\leq k$, the  {\it$s$-maximum Hamming weight} of $C$ is defined as $$D_{s}(C)=\max\{|\mathrm{supp}(U)|\,\big|\, U\in {\rm SUB}^{s}(U) \}.$$
\end{Definition}

When $s=1,$ $D_{1}(C)$ is   the maximum Hamming distance of $C$.
The notion of the subcode support weight distributions of linear codes is provide in the following definition.
\begin{Definition}
Let $C$ be an $[n,k]_q$-linear code,
the sequence $[A_1^{s}(C), A_2^{s}(C),\cdots, A_n^{s}(C)]$ is called the {\it $s$-subcode support weight distribution} ($s$-SSWD) of $C$,
where $$A_j^{s}(C)=|\{ U\,\big{|}\,U\text{ is a subspace of $C$ with dimension $s$ and $w(U)=j$} \}|,$$ for $1\leq j\leq n$ and $1\leq s\leq k.$
\end{Definition}
Note that $A_j(C)=(q-1)A_j^1(C)$  for $ 1\leq j\leq n$.
 Let $G =[G^T_{1},G^T_{2},\cdots,G^T_{n}]$ be a $k\times n$ matrix over $\mathbb{F}_q$,  where $G_i\in\mathbb{F}_q^{k}$ for all $1\leq i\leq n$. For any $V\in {\rm SUB}(\mathbb{F}_q^{k})$, the function $m_{G}: {\rm SUB}(\mathbb{F}_q^{k}) \to \mathbb{N}$ is defined as following:
$$
m_{G}(V)=|\{i\in [n]\,\big{|}\, G_{i} \in V\}|.
$$

\begin{Remark}
Let $G$ and $\tilde{G} $ be two matrices of the same column size such that the columns of $G$ are contained in those columns of $\tilde{G}$. The matrix obtained by puncturing the columns of $G$ from $\tilde{G}$ is denoted by $\tilde{G} \backslash G$.
Then we know that $m_{\tilde{G}}(V)=m_{G}(V)+m_{\tilde{G} \backslash G }(V)$ for any $V\in {\rm SUB}^s(\mathbb{F}_q^{k}) .$
\end{Remark}

Let $C$ be an $[n,k]_q$-linear code with a generator matrix $G$. Then we know that, for any subspace $U\in {\rm SUB}^{s}(C)$, there exists a subspace $V\in {\rm SUB}^{s}(\mathbb{F}_q^k)$ such that $$U=\{ \mathbf{y}G\,|\,  \mathbf{y}\in V\}.$$
The following lemma is easy to be obtained.
\begin{lem} \label{weight}
Assume the notation is as given above.
For an $[n,k]_q$-linear code $C$ with a generator matrix $G=[G^T_{1},G^T_{2},\cdots,G^T_{n}]$ and a subspace $U\in {\rm SUB}^{s}(C)$,
 the support weight of $U$ is $$w(U)= n-m_{G}(V ^{\bot}).$$
Then we have follows:
\begin{description}
  \item[(a)]  The $s$-generalized Hamming weight of $C$ is $$d_s(C)=n-\max\{m_G(V)\,|\, V\in {\rm SUB}^{k-s}(\mathbb{F}_q^{k})\}$$
for $1\leq s\leq k.$
  \item[(b)] The $s$-maximum Hamming weight of $C$ is $$D_s(C)=n-\min\{m_G(V)\,|\, V\in {\rm SUB}^{k-s}(\mathbb{F}_q^{k})\}$$
for $1\leq s\leq k.$

  \item[(c)] The subcode support weight distributions of $C$ satisfies
$$A^s_{j}(C)=|\{ V\in  {\rm SUB}^{k-s}(\mathbb{F}_q^k)\,\big{|}\, m_{G}(V)=n-j \}|$$ for $1\leq s\leq k$ and $1\leq j\leq n.$

\end{description}
\end{lem}

\section{Some new lower bound for the size of $t$-fold $s$-blocking sets}
Let $\mathrm{PG}(k-1,q)$ be the
finite projective geometry of dimension $k-1$ over  $\mathbb{F}_q^{k}$.
Then there is a one-to-one correspondence between all the points of $\mathrm{PG}(k-1,q)$ and
${\rm SUB}^{1}(\mathbb{F}_q^k).$
In the following theorem, we provide some new lower bound for the size of $t$-fold $s$-blocking sets. For  a  $t$-fold $s$-blocking set $B$ of $\mathrm{PG}(k-1,q)$, the complementary set of $B$ in $\mathrm{PG}(k-1,q)$ is
denoted by $\mathrm{PG}(k-1,q)\backslash B$.
\begin{Theorem}\label{y1}
 Let $B$ be a  $t$-fold $s$-blocking set of $\mathrm{PG}(k-1,q)$ with $1\leq s\leq k-1$.
\begin{description}
  \item[(a)]   Then  $|B|\ge t\frac{q^k-1}{q^{k-s}-1}. $

   \item[(b)] If $\mathrm{PG}(k-1,q)\backslash B$ does not generate $\mathrm{PG}(k-1,q),$
  then $|B|\ge q^{k-1}  .$

  \item[(c)] If $\mathrm{PG}(k-1,q)\backslash B$ generates $\mathrm{PG}(k-1,q),$ then $$ |B| \ge \min\{  t\frac{q^{s+1}-1}{q-1} ,\, t+\frac{q^2(q^{s}-1)}{q-1}\}.$$
In particular,  when $t\leq q,$ we have that $|B| \ge t\frac{q^{s+1}-1}{q-1}.$

\end{description}
\end{Theorem}

\begin{proof}
{\bf (a)} Suppose that $B$ is a  $t$-fold $s$-blocking set of $\mathrm{PG}(k-1,q)$  of size $n$.
Let $\{G_1, G_2, \cdots, G_n\}$ be the subset of  $\mathbb{F}_q^k$ corresponding to  $B$.
 Let $\Lambda=\{ (i, V)\,|\, V\in {\rm SUB}^{k-s}( \mathbb{F}_q^k)\text{ and } G_i\in V\}.$
 Then we know that
 \begin{eqnarray*}
   |\Lambda | &=&   \sum_{i=1}^n | \{  V\in {\rm SUB}^{k-s}( \mathbb{F}_q^k)\,\big{|}\, G_i\in V\}|\\
     &=&   n\dst{k-1}{k-s-1}_q .
 \end{eqnarray*}

Since $| B\bigcap V |\ge t$ for every $V\in {\rm SUB}^{k-s}(\mathbb{F}_q^k),$  we have that
\begin{eqnarray*}
   |\Lambda | &=&   \sum_{ V\in {\rm SUB}^{k-s}( \mathbb{F}_q^k)} | \{ i\in [n]\,\big{|}\, G_i\in V\}|\\
     &=&  \sum_{ V\in {\rm SUB}^{k-s}( \mathbb{F}_q^k)} | B\bigcap V | \\
     &\ge & t\dst{k}{k-s}_q .
 \end{eqnarray*}
Hence we get that $$|B|=n\ge t\frac{\dst{k}{k-s}_q}{\dst{k-1}{k-s-1}_q}=t\frac{q^k-1}{q^{k-s}-1}  .$$

{\bf (b)} Since $\mathrm{PG}(k-1,q) \backslash B$ does not generate $\mathrm{PG}(k-1,q),$ there exists a hyperplane $U$ such that $\mathrm{PG}(k-1,q) \backslash B \subseteq U .$
Then $\mathrm{PG}(k-1,q) \backslash U\subseteq  B$
and $$|B|\ge \frac{q^k-1}{q-1}  -\frac{q^{k-1}-1}{q-1}= q^{k-1}  .$$

{\bf (c)}
Let $\{G_1, G_2, \cdots, G_N\}$ be the subset of  $\mathbb{F}_q^k$ corresponding to  $\mathrm{PG}(k-1,q) \backslash B$, where $N= \frac{q^k-1}{q-1} -|B|.$
Let $C$ be the linear code generated by the $k\times N$ matrix
$$G =[G^T_{1},G^T_{2},\cdots,G^T_{N}].$$
Since $\mathrm{PG}(k-1,q) \backslash B$ generates $\mathbb{F}_q^k,$
we know that $C$ is an $[N,k]_q$-linear code.
For any $U\in {\rm SUB}^{s}(C),$
there exist subspaces $V\in {\rm SUB}^{s}( \mathbb{F}_q^k)$ such that $$U=\{ \mathbf{y}G\,|\,  \mathbf{y}\in V\}.$$
By Lemma~\ref{weight}, we know that
\begin{eqnarray*}
  w(U) &=&n-m_G(V^{\bot})  \\
   &=& \frac{q^k-1}{q-1} -|B|-(\frac{q^{k-s}-1}{q-1}- |B \bigcap V^{\bot}|)  \\
   &=&  \frac{q^k-q^{k-s}}{q-1} -|B| +|B \bigcap V^{\bot}|\\
   &\ge &   \frac{q^k-q^{k-s}}{q-1} -|B| +t
\end{eqnarray*}
and $d_s (C) \ge    \frac{q^k-q^{k-s}}{q-1} -|B| +t .$
By the Griesmer bound on the generalized Hamming weights, we have that
\begin{eqnarray*}
  N  &\ge &d_s(C)+\sum_{i=1}^{k-s}\lceil \frac{(q-1)d_s(C)}{q^i(q^s-1)} \rceil \\
    &\ge & \frac{q^k-q^{k-s}}{q-1} -|B| +t +\sum_{i=1}^{k-s}\lceil \frac{q^k-q^{k-s}-(|B| -t)(q-1)}{q^i(q^s-1)} \rceil\\
    &\ge &  \frac{q^k-q^{k-s}}{q-1} -(|B| -t) +\sum_{i=1}^{k-s}\lceil q^{k-s-i}- \frac{(|B| -t)(q-1)}{q^i(q^s-1)}\rceil \\
    &=&   \frac{q^k-q^{k-s}}{q-1} -(|B| -t) + \sum_{i=1}^{k-s} q^{k-s-i}- \sum_{i=1}^{k-s} \lfloor \frac{(|B| -t)(q-1)}{q^i(q^s-1)} \rfloor\\
    &=&   \frac{q^k-1}{q-1}  -(|B| -t) - \sum_{i=1}^{k-s} \lfloor \frac{(|B| -t)(q-1)}{q^i(q^s-1)} \rfloor .
\end{eqnarray*}
Since $N= \frac{q^k-1}{q-1} -|B|$, we know that $$t\leq \sum_{i=1}^{k-s} \lfloor \frac{(|B| -t)(q-1)}{q^i(q^s-1)} \rfloor \leq  \frac{(|B| -t)(q-1)}{q(q^s-1)}+\sum_{i=2}^{k-s} \lfloor \frac{(|B| -t)(q-1)}{q^i(q^s-1)} \rfloor.$$
 We analysis the following two cases.

 {\em Case 1}. Suppose that $\frac{(|B| -t)(q-1)}{q^2(q^s-1)} <1,$
 which means $\sum\limits_{i=2}^{k-s} \lfloor \frac{(|B| -t)(q-1)}{q^i(q^s-1)} \rfloor=0.$ Then we know that $$t\leq  \frac{( |B|-t)(q-1)}{q(q^s-1)}$$ and $$|B| \ge t\frac{q^{s+1}-1}{q-1}.$$

 {\em Case 2}. Suppose that $\frac{(|B| -t)(q-1)}{q^2(q^s-1)} \ge 1,$  which means $$ |B| \ge t+\frac{q^2(q^{s}-1)}{q-1}.$$
 Hence $ |B| \ge \min\{  t\frac{q^{s+1}-1}{q-1} ,\, t+\frac{q^2(q^{s}-1)}{q-1}\}.$

 In particular,  when $t\leq q,$ we have that $$t+\frac{q^2(q^{s}-1)}{q-1}\ge t+t\frac{q(q^{s}-1)}{q-1} = t\frac{q^{s+1}-1}{q-1} $$
 and $$|B| \ge t\frac{q^{s+1}-1}{q-1}.$$

\end{proof}

\begin{Remark}
 Let $B$ be a  $t$-fold $s$-blocking set of $\mathrm{PG}(k-1,q)$ with $1\leq s\leq k-1$.
Theorem 4 of \cite{Beu} prove that,  if  $t\leq q, $  then $|B|\ge t\frac{q^{s+1}-1}{q-1}.$ But, when $t> q,$ our bounds in Theorem~\ref{y1} still can be used.
\end{Remark}

\section{Characterization of $s$-minimal codes as cutting $s$-blocking sets}

In this section, we provide a connection between $s$-minimal codes and  cutting $s$-blocking sets.
Also,  we show that, if $C$ is a  $(s+1)$-minimal code, then $C$ is a  $s$-minimal code.
In the following definition, we introduce the notion of $s$-minimal subcode and $s$-minimal codes.
\begin{Definition}
Let $C$ be  an $[n,k]_q$-linear code.  For $U\in {\rm SUB}^{s}(C)$,
$U$ is called a {\it $s$-minimal subcode} of $C$, if there exists a $V\in {\rm SUB}^{s}(C)$ such that $\mathrm{supp}(V)\subseteq \mathrm{supp}(U)$ then  $V=U.$
 If
$U$ is a $s$-minimal subcodes of $C$ for every $U\in {\rm SUB}^{s}(C)$, then $C$ is called a $s$-minimal   code.
\end{Definition}

When $s=1$,  a $s$-minimal linear  code is also called a {\it minimal code} \cite{AB}.

A cutting blocking set of $\mathrm{PG}(k-1,q)$ of size $n$ corresponds to
a minimal $[n, k]_q$-linear code in \cite{AB20,HN,TQ}.
In the following theorem, we prove that a cutting $s$-blocking set of $\mathrm{PG}(k-1,q)$ of size $n$ corresponds to
a projective $s$-minimal $[n, k]_q$-linear code.
\begin{Theorem}\label{yy}
  Let $s$ and $k$ be positive integers such that $1\leq s\leq k-1.$
There exists a  cutting $s$-blocking set of $\mathrm{PG}(k-1,q)$ of size $n$ if and only if there exists  a projective $s$-minimal $[n, k]_q$-linear code.
\end{Theorem}
\begin{proof}
$(\Rightarrow)$ Suppose that $B$ is a  cutting $s$-blocking set of $\mathrm{PG}(k-1,q)$  of size $n$.
Let $\{G_1, G_2, \cdots, G_n\}$ be the subset of  $\mathbb{F}_q^k$ corresponding to  $B$, where $\mathbf{0}\neq G_i\in\mathbb{F}_q^{k}$ for all $1\leq i\leq n$.
Let $C$ be an $[n, k]_q$-linear code generated by the $k\times n$ matrix
$$G =[G^T_{1},G^T_{2},\cdots,G^T_{n}].$$
Since $\langle G_i\rangle \neq \langle G_j\rangle$ for $i\neq j,$
we know that $C$ is projective.

Let $U_1$ and $U_2$ be in ${\rm SUB}^{s}(C)$ such that $\mathrm{supp}(U_1)\subseteq \mathrm{supp}(U_2).$
There exist subspaces $V_1$ and $V_2$ in ${\rm SUB}^{s}( \mathbb{F}_q^k)$ such that $$U_1=\{ \mathbf{y}G\,|\,  \mathbf{y}\in V_1\} \text{ and } U_2=\{ \mathbf{y}G\,|\,  \mathbf{y}\in V_2\}.$$
Let $\tilde{V}_j^{\bot}$ be the projective subspace of $\mathrm{PG}(k-1,q)$ of codimension $s$ corresponding to $V_j^{\bot}$
for $j=1,2.$
Since  $ [n] \backslash \mathrm{supp}(U_j)= \{ i\in [n]\,|\,   G_i \in   V_j^{\bot} \}$ for $j=1,2$
and $ [n] \backslash \mathrm{supp}(U_2) \subseteq [n] \backslash \mathrm{supp}(U_1),$
we have that  $$ B\bigcap  \tilde{V}_2^{\bot} \subseteq  B\bigcap  \tilde{V}_1^{\bot}  .$$
Since  $B$ is a  cutting $s$-blocking set,
we get that $$ \tilde{V}_2^{\bot}=\langle B\bigcap  \tilde{V}_2^{\bot} \rangle \subseteq  \langle B\bigcap  \tilde{V}_1^{\bot}\rangle =  \tilde{V}_1^{\bot}$$
and $ \tilde{V}_1^{\bot}=\tilde{V}_2^{\bot}.$ Then $V_1=V_2$ and $U_1=U_2$.
Hence $C$ is a projective $s$-minimal $[n, k]_q$-linear code.

$(\Leftarrow)$ Suppose that $C$  is a projective $s$-minimal $[n, k]_q$-linear code with a generator matrix
$$G =[G^T_{1},G^T_{2},\cdots,G^T_{n}],$$   where $\mathbf{0}\neq G_i\in\mathbb{F}_q^{k}$ for all $1\leq i\leq n$.
Since $C$ is projective,
we know that $\langle G_i\rangle \neq \langle G_j\rangle$ for $i\neq j$.
Let $B$ be a set of $\mathrm{PG}(k-1,q)$  of size $n$  corresponding to $\{G_1, G_2, \cdots, G_n\}$.
Then the size of $B$ is $n$.

Suppose that there exists a   projective subspace $\tilde{V}_1$ of $\mathrm{PG}(k-1,q)$ of codimension $s$
such that $\tilde{V}_1$ is not generated by $B \bigcap \tilde{V}_1.$
Then there exists a  projective subspace $\tilde{V}_2$ of $\mathrm{PG}(k-1,q)$ of codimension $s$
such that $$ B \bigcap V_1 \subsetneqq  B \bigcap V_2.$$

Let $V_j\in {\rm SUB}^{k-s}( \mathbb{F}_q^k)$ such that $V_j$ corresponding  $\tilde{V}_j$ for $j=1,2.$
Let $U_j=\{ \mathbf{y}G\,|\,  \mathbf{y}\in V_j^{\bot}\}$ for $j=1,2.$
Then we know that $U_1$ and $U_2$ are in ${\rm SUB}^{s}(C)$ such that $$\mathrm{supp}(U_1)\subsetneqq \mathrm{supp}(U_2),$$ which is a contradiction.
Hence $B$ is a  cutting $s$-blocking set of $\mathrm{PG}(k-1,q)$ of size $n$.
\end{proof}

For an $[n, k]_q$-linear code $C$ with a generator matrix $G =[G^T_{1},G^T_{2},\cdots,G^T_{n}],$
we construct a linear code $\tilde{C}$ by puncturing on all the $j$-th coordinates of $C$ such that there exists $i$ with $1\leq i <j$ satisfying $G_{i}=\alpha G_{j}$  for some $\alpha \in \mathbb{F}_q^{*}.$
The linear code $\tilde{C}$ is called a  {\it projection} of $C$.
It is easy to get that the linear code $\tilde{C}$ is projective.
In the following theorem, we prove that $C$ is a  $s$-minimal code if and only if $\tilde{C}$ is a  $s$-minimal code.

 \begin{Theorem}\label{yyy}
Let $C$ be an $[n, k]_q$-linear code  with a projection $\tilde{C}$ for $1\leq s\leq k-1$. Then $C$ is a  $s$-minimal code if and only if $\tilde{C}$ is a  $s$-minimal code.
\end{Theorem}
\begin{proof}
Let $G =[G^T_{1},G^T_{2},\cdots,G^T_{n}]$ be a generator matrix  of $C$, and let $\tilde{G} =[\tilde{G}^T_{1},\tilde{G}^T_{2},\cdots,\tilde{G}^T_{m}]$ be a generator matrix  of $\tilde{C}.$
Assume that $V_1$ and $V_2$ are in ${\rm SUB}^{s}( \mathbb{F}_q^k).$

Let  $U_i=\{ \mathbf{y}G\,|\,  \mathbf{y}\in V_i\} $ and $\tilde{U}_i=\{ \mathbf{y}\tilde{G}\,|\,  \mathbf{y}\in V_i\}$ for $i=1,2.$
If $\mathrm{supp}(U_1)\subseteq \mathrm{supp}(U_2),$
then we have that  $\mathrm{supp}(\tilde{U}_1)\subseteq \mathrm{supp}(\tilde{U}_2) .$

By the definition of $\tilde{C}$, for any $1\leq i \leq n,$ there exists $j$  with $1\leq j\leq m $ such that $G_{i}=\alpha\tilde{G}_{j} $ for some $\alpha \in \mathbb{F}_q^{*}.$
Therefore, if $\mathrm{supp}(\tilde{U}_1)\subseteq \mathrm{supp}(\tilde{U}_2) ,$
then we have that  $\mathrm{supp}(U_1)\subseteq \mathrm{supp}(U_2).$
Then $U_2$ is a $s$-minimal subcode of $C$  if and only if  $\tilde{U}_2$ is a $s$-minimal subcode of $\tilde{C}$.
Hence $C$ is a  $s$-minimal code if and only if $\tilde{C}$ is a  $s$-minimal code.
\end{proof}

Recall that  $m(k,s,q)$ is the minimal length $n$ such that there exists a $s$-minimal $[n, k]_q$-linear code.
By Theorem~\ref{yyy}, we  have that $m(k,s,q)$ is the minimal length $n$ such that there exists a projective $s$-minimal $[n, k]_q$-linear code.
By Theorem~\ref{yy}, we  have that $m(k,s,q)$ is the minimal number $n$ such that there exists a  cutting $s$-blocking set of $\mathrm{PG}(k-1,q)$ of size $n$.
In the following theorem, we provide some  lower bounds for $m(k,s,q)$ by Theorem~\ref{y1}.

\begin{Corollary}\label{y4}
Assume that the notation is as given above.
\begin{description}
  \item[(a)]
Then $m(k,s,q)\ge (k-s)\frac{q^k-1}{q^{k-s}-1}.$
\item[(b)]   If  $k-s\leq q, $  then $m(k,s,q)\ge (k-s)\frac{q^{s+1}-1}{q-1}. $
\end{description}
\end{Corollary}

\begin{proof}
Note that, if $B$ is a  cutting $s$-blocking set of $\mathrm{PG}(k-1,q)$ of size $n$, then $B$ is a
 $(k-s)$-fold $s$-blocking set of $\mathrm{PG}(k-1,q).$
It is a direct result by Theorem~\ref{y1} and Theorem~\ref{yy}.
\end{proof}

In the following theorem, we show that, if $C$ is a  $(s+1)$-minimal code, then $C$ is a  $s$-minimal code.

\begin{Theorem}\label{T5}
 Let $C$ be an  $[n,k]_q$-linear code and $1\leq s\leq k-2.$

\begin{description}
  \item[(a)] If $C$ is a  $(s+1)$-minimal code, then $C$ is a  $s$-minimal code.
\item[(b)]  If $B$ is a  cutting $(s+1)$-blocking set of $\mathrm{PG}(k-1,q),$
then $B$ is a  cutting $s$-blocking set of $\mathrm{PG}(k-1,q).$
\end{description}
\end{Theorem}

\begin{proof}
 {\bf (a)} Suppose that $C$ is not a  $s$-minimal code.
  There there exist $U_1,\,U_2\in  {\rm SUB}^{s}(C)$ such that $\mathrm{supp}(U_1) \subseteq \mathrm{supp}(U_2)$ and
  $\mathrm{supp}(U_1) \neq \mathrm{supp}(U_2).$
  Then there exist a integer $i_0$
  and a codeword $\mathbf{x}\in U_2$ such that $i_0\in  \mathrm{supp}(\mathbf{x}) \subseteq \mathrm{supp}(U_2) $ and $i_0\notin\mathrm{supp}(U_1).$
  Since $\mathbf{x}\notin U_1,$ we know that $\dim(\langle\mathbf{x}, U_1\rangle)=s+1.$

  Since $k\ge s+2$ and $q^k>q^{s+1}+q^s,$ there exists $\mathbf{z}\in C$ such that $$\mathbf{z}\notin \langle\mathbf{x}, U_1\rangle \,\text{ and }\,\mathbf{z}\notin U_2.$$ Then we know that
  $$\dim(\langle\mathbf{z}, U_1\rangle)=\dim(\langle\mathbf{z}, U_2\rangle)=s+1.$$
  Note that $\mathrm{supp}(\langle\mathbf{z}, U_1\rangle) \subseteq \mathrm{supp}(\langle\mathbf{z}, U_2\rangle)$.
  Since $C$ is a  $(s+1)$-minimal code, we have that $$\mathrm{supp}(\langle\mathbf{z}, U_1\rangle) = \mathrm{supp}(\langle\mathbf{z}, U_2\rangle)$$
  and $$i_0\in \mathrm{supp}(U_2) \subseteq \mathrm{supp}(\langle\mathbf{z}, U_1\rangle).$$
  Therefore  $i_0\in \mathrm{supp}(\mathbf{z}).$ Since $i_0\in \mathrm{supp}(\mathbf{x}),$
  there exists a nonzero element $\alpha \in \mathbb{F}_q$ such that $i_0\notin \mathrm{supp}(\mathbf{x}+\alpha\mathbf{z}).$

  Note that $\mathrm{supp}(\langle\mathbf{x}+\alpha\mathbf{z}, U_1\rangle) \subseteq \mathrm{supp}(\langle\mathbf{x}+\alpha\mathbf{z}, U_2\rangle)$.
By $\mathbf{z}\notin \langle\mathbf{x}, U_1\rangle$ and $\mathbf{x}\in U_2$, we have that
 $$\dim(\langle\mathbf{x}+\alpha\mathbf{z}, U_1\rangle)=\dim(\langle\mathbf{x}+\alpha\mathbf{z}, U_2\rangle)=s+1.$$
 Since $C$ is a  $(s+1)$-minimal code, we have that $$\mathrm{supp}(\langle\mathbf{x}+\alpha\mathbf{z}, U_1\rangle) = \mathrm{supp}(\langle\mathbf{x}+\alpha\mathbf{z}, U_2\rangle)$$
  and $$i_0\in \mathrm{supp}(U_2) \subseteq \mathrm{supp}(\langle\mathbf{x}+\alpha\mathbf{z}, U_1\rangle).$$
  But we know that $i_0\notin\mathrm{supp}(U_1)$ and $i_0\notin \mathrm{supp}(\mathbf{x}+\alpha\mathbf{z}),$
  which is a contradiction.
 Hence $C$ is a  $s$-minimal code.

 {\bf (b)} It is a direct result of Statement {\bf (a)} and Theorem~\ref{yy}.
\end{proof}

\begin{Remark}
In Example~\ref{ee12pa}, we provide a $[28,5,14]_2$-linear code $C_1$ such that $C_1$  is a  $s$-minimal code for $s=1,2$ but $C_1$  is not a  $s$-minimal code for $s=3,4.$
\end{Remark}

\begin{Example}{\rm
In \cite[Page 8]{CK}, a doubly-even minimal $[29,8,12]_2$ code with the weight enumerator $1+114x^{12}+119x^{16}+22x^{20}$ was given. From the Magma calculation, this code is not a 2-minimal code.
In addition, an $8$-divisible minimal $[42,7,16]_2$ code with the weight enumerator $1+45x^{16}+82x^{24}$ was presented in \cite[Page 8]{CK}. By the Magma calculation, this code is not 2-minimal.
}
\end{Example}

\begin{Example}{\rm
    Let $q=2$ and $n=85$.
    Let the defining set of a cyclic code of length $85$ to be the union of all $2$-cyclotomic cosets modulo  $85$ except the one with coset leader $37$, we obtain a cyclic code $C$ with parameters $[85,8,40]_2$.
    Using Magma, we found that this code $C$ is a two-weight cyclic code with nonzero weights $40$ and $48$. This code $C$ is $s$-minimal for $s=1,2$, but not a $3$-minimal code.
    On the other hand, the optimal linear $[85,8,40]_2$ code $C_1$ in \cite{G} has nonzero weights $40$, $48$ and $56$. This code $C_1$ is minimal but not $2$-minimal.}
\end{Example}

By Theorem~\ref{T5}, we obtain the following lower bound for the minimum Hamming distance of $s$-minimal codes.
\begin{Corollary}\label{yt}
Let $C$ be a $s$-minimal $[n,k,d]_q$-linear code with $1\leq s\leq k-1$. Then
 $$d\ge (q-1)(k-1)+1 .$$
\end{Corollary}
\begin{proof}
By Theorem~\ref{T5}, we know that $C$ is a minimal code.
By Theorem 2.8 of \cite{AB22}, we get that  $d\ge (q-1)(k-1)+1 .$
\end{proof}

\section{$s$-Minimal subcodes of  linear codes}
In this section, we provide some properties of $s$-minimal subcodes,  and give some sufficient conditions of  that a  subcode is a $s$-minimal subcode.
For   an $[n,k]_q$-linear code $C$, let $H$ be the parity-check matrix of $C$. The submatrix of $H$
consisting columns indexed by a subset $X\subseteq[n]$, is denoted  by $H(X)$.
The basic properties of  $s$-minimal subcodes are characterized in the following lemma.
 \begin{lem} \label{L1}
Let $C$ be an $[n,k]_q$-linear code with the generalized Hamming weight hierarchy $\{d_{1}(C),d_{2}(C),\cdots,d_{k}(C)\}$,
and let $H$ be a parity-check matrix of $C$.
\begin{description}
  \item[(a)] For $U\in {\rm SUB}^{s}(C)$, $ U$ is
a $s$-minimal subcode of $C$ if and only if the rank of $H(\mathrm{supp}(U))$ is $|\mathrm{supp}(U)|-s.$

\item[(b)] If $ U$ is
a $s$-minimal subcode of $C$, we have that $|\mathrm{supp}(U)|\leq n-k+s.$
\end{description}
\end{lem}
\begin{proof}
{\bf (a)} Suppose that $ U$ is a $s$-minimal subcode of $C$.
Let $G_U$ be a generator matrix of $U.$
Then we know that $$G_U(\mathrm{supp}(U))H^T(\mathrm{supp}(U))=\mathbf{0},$$ and the rank of $G_U(\mathrm{supp}(U))$ is equal to $s$. Hence the rank of $H(\mathrm{supp}(U))$ is less than or equal to $|\mathrm{supp}(U)|-s.$
If the rank of $H(\mathrm{supp}(U))$ is less than  $|\mathrm{supp}(U)|-s,$
then there exists  $V\in {\rm SUB}^{s}(C)$ such that $U\neq V$ and $\mathrm{supp}(V)\subseteq \mathrm{supp}(U) ,$
which is a contradiction. Hence the rank of $H(\mathrm{supp}(U))$ is  equal to $|\mathrm{supp}(U)|-s.$

Suppose the rank of $H(\mathrm{supp}(U))$ is  equal to $|\mathrm{supp}(U)|-s.$
For any $V\in {\rm SUB}^{s}(C)$ such that $\mathrm{supp}(V)\subseteq \mathrm{supp}(U) ,$
let $G_V$ be a generator matrix of $V.$ Then we have that $$G_V(\mathrm{supp}(U))H^T(\mathrm{supp}(U))=\mathbf{0}$$
Since the rank of $G_U(\mathrm{supp}(U))$ and the rank of $G_V(\mathrm{supp}(U))$ are both equal to $s$, we have that $V=U.$

 {\bf (b)} Suppose that $ U$ is a $s$-minimal subcode of $C$.
 Since the rank of $H(\mathrm{supp}(U))$ is less than or equal to $n-k,$
 we have that $|\mathrm{supp}(U)|\leq n-k+s$ by Statement {\bf (a)}.
\end{proof}

 In the following theorem, we  generalize Ashikhminand-Barg theorem \cite{AB} from codewords to subcodes of a linear code.
 \begin{Theorem}\label{T1}
Let $C$ be an  $[n,k]_q$-linear code with the generalized Hamming weight hierarchy $\{d_{1}(C),d_{2}(C),\cdots,d_{k}(C)\}$. Assume that  $1\leq s\leq k-1$.
\begin{description}
  \item[(a)] If $U\in {\rm SUB}^{s}(C)$ such that $|\mathrm{supp}(U)|<d_{s+1}(C),$ then $U$ is a $s$-minimal subcode of $C$.
  \item[(b)]  If $U\in {\rm SUB}^{s}(C)$ such that $$ |\mathrm{supp}(U)|<\frac{q^{s+1} -1}{q^{s+1}-q}d_{s}(C),$$ then $U$ is a $s$-minimal subcode of $C$.
\end{description}
\end{Theorem}

\begin{proof}
{\bf (a)} Let $V\in {\rm SUB}^{s}(C)$ such that $\mathrm{supp}(V)\subseteq \mathrm{supp}(U) .$
Suppose that $V\neq U.$
Then the dimension of the subspace generated by $V$ and $U$ is great than or equal to $s+1.$
Therefore we have that $$ d_{s+1}(C) \leq |\mathrm{supp}(U) |< d_{s+1}(C),$$
which is a contradiction.
Hence $V= U.$

{\bf (b)}
By Lemma~\ref{L2}, we know that $|\mathrm{supp}(U)|<d_{s+1}(C)$.
By Statement {\bf (a)}, we know that $U$ is a $s$-minimal subcode of $C$.
\end{proof}

\begin{Example}\label{ee132q}
Let $C$ be the $[9, 5, 3]_{2}$-linear code with the generator matrix
$$\left(\begin{array}{cccc cccc c}
1&0&0&0&1&0&1&1&1\\
0&1&0&0&1&0&1&1&0\\
0&0&1&0&1&0&1&0&1\\
0&0&0&1&1&0&0&1&1\\
0&0&0&0&0&1&1&1&1
\end{array}\right) .$$
Then  we know that $d_{1}(C)=3$, $d_{2}(C)=5$, $d_{3}(C)=7$, $d_{4}(C)=8$ and $d_{5}(C)=9$.
By Magma, the subcode support weight distributions of $C$ are listed in
Table~\ref{ta71ww}.
For example, when $U\in {\rm SUB}^{2}(C)$, we know that $U$ is a $2$-minimal subcode of $C$ if and only if $|\mathrm{supp}(U)|<d_{3}(C)=7$.

\begin{table}[htbp]
\caption{The subcode support weight distributions of the linear code in Example~\ref{ee132q}}
\label{ta71ww}
\center

{\tiny
\begin{tabular}{c|c|c |c|c|c|c|c}
&$A^r_{3}(C)$&$A^r_{4}(C)$ &$A^r_{5}(C)$&$A^r_{6}(C)$ &$A^r_{7}(C)$ &$A^r_{8}(C)$ &$A^r_{9}(C)$\\
\hline
$r=1$ &4 (minimal)& 14 (minimal)&8 (minimal)& 0& 4(not minimal)&1(not minimal)&0\\\hline
 $r=2$ &- &- &6 (minimal)&60 (minimal)&36(not minimal)&39(not minimal)&14\\\hline
 $r=3$ &- & -&-&-&36 (minimal)&63(not minimal)&56\\\hline
 $r=4$ &- &- &-& -&- &9 (minimal)&22\\\hline
 $r=5$ & -&- &-&-&-&-&1\\
\end{tabular}
}
\end{table}
\end{Example}

  \begin{Remark}
  In  Theorem~\ref{T1}, the condition of $$ |\mathrm{supp}(U)|<\frac{q^{s+1} -1}{q^{s+1}-q}d_{s}(C)$$ is necessary for  that $U$ is a $s$-minimal subcode of some linear code $C$. For example,  $C$ is an $[n,k]_q$-linear code such that $n-k+s<\frac{q^{s+1} -1}{q^{s+1}-q}d_{s}(C)$.
  This result is proved in the following theorem.
  \end{Remark}

   \begin{Theorem}\label{T1a}
 Let   $C$ be  an $[n,k]_q$-linear code  and $U\in {\rm SUB}^{s}(C)$.
 \begin{description}
   \item[(a)] Assume  $d_{s+1}(C)=n-k+s+1$.
Then $U$ is a $s$-minimal subcode of $C$ if and only if  $|\mathrm{supp}(U)|<d_{s+1}(C)$.
   \item[(b)]  Assume $n-k+s<\frac{q^{s+1} -1}{q^{s+1}-q}d_{s}(C) .$
Then $U$ is a $s$-minimal subcode of $C$ if and only if  $ |\mathrm{supp}(U)|<\frac{q^{s+1} -1}{q^{s+1}-q}d_{s}(C)$.
 \end{description}
   \end{Theorem}
  \begin{proof}
   {\bf (a)} By   Theorem~\ref{T1}, we only need to prove that, if  $U$ is  a $s$-minimal subcode, then  $|\mathrm{supp}(U)|<d_{s+1}(C)$ .
 Suppose that $U$ is  a $s$-minimal subcode.
 By Lemma~\ref{L1}, we know that $$|\mathrm{supp}(U)|\leq n-k+s <d_{s+1}(C)=n-k+s+1 .$$

     {\bf (b)} By   Theorem~\ref{T1}, we only need to prove that, if  $U$ is a $s$-minimal subcode, then $ |\mathrm{supp}(U)|\ge \frac{q^{s+1}-1}{q^{s+1}-q}d_{s}(C).$
  Suppose that $U$ is  a $s$-minimal subcode.
By Lemma~\ref{L1}, we know that $$|\mathrm{supp}(U)|\leq n-k+s < \frac{q^{s+1} -1}{q^{s+1}-q}d_{s}(C). $$
  \end{proof}

  \begin{Remark}
 {\bf (a)} There are several infinite families of $[n,k]_q$-linear codes $C$ satisfying the condition $d_{s+1}(C)=n-k+s+1$  of  Theorem~\ref{T1a}.
  For example, if $C$ is a  maximum distance separable (MDS) code, which  means $d_{1}(C)=n-k+1$,
  then $d_{s}(C)=n-k+s$ for every $1\leq s\leq k$.
   if $C$ is a near  maximum distance separable (NMDS) code, which  means $d_{1}(C)=n-k$ and $d_{1}(C^{\bot})=k$,
  then $d_{s}(C)=n-k+s$ for every $2\leq s\leq k$.

   {\bf (b)} There are several infinite families of $[n,k]_q$-linear codes $C$ satisfying the condition $n-k+s< \frac{q^{s+1} -1}{q^{s+1}-q}d_{s}(C)$  of  Theorem~\ref{T1a}. If $C$ is   an $[n,k]_q$-linear code such that $d_{s}(C)=n-k+s,$
   we know that $n-k+s< \frac{q^{s+1} -1}{q^{s+1}-q}d_{s}(C)$ since $\frac{q^{s+1} -1}{q^{s+1}-q}>1.$
   \end{Remark}

  \begin{Example}\label{ew1}
Let $C$ be the $[12,5,6]_5$-linear code with the generator matrix
$$\left(\begin{array}{ccccc ccccc cc}
1&0&0&0&2&3&0&1&4&3&0&1\\
0&1&0&0&1&3&0&1&3&1&4&4\\
0&0&1&0&1&2&0&1&2&0&2&0\\
0&0&0&1&1&1&0&0&0&1&1&3\\
0&0&0&0&0&0&1&1&1&1&1&1
\end{array}\right) .$$
Then the dual code $C^{\bot}$ is a $[12,7,5]_5$-linear code.
 By magma, the subcode support weight distributions of $C$ are listed in
Table~\ref{ta71ww3}.
 Then $C$ satisfies $d_{s}(C)=n-k+s$ and
 $n-k+s<\frac{q^{s+1} -1}{q^{s+1}-q}d_{s}(C) $, when $s=2.$
 For example, when $U\in {\rm SUB}^{1}(C)$, we know that $U$ is a $1$-minimal subcode of $C$ if and only if $|\mathrm{supp}(U)|<9$.

\begin{table}[htbp]
\caption{The subcode support weight distributions of the linear code in Example~\ref{ew1}}
\label{ta71ww3}
\center

{\tiny
\begin{tabular}{c|c|c |c|c|c|c|c}
&$A^r_{6}(C)$&$A^r_{7}(C)$ &$A^r_{8}(C)$&$A^r_{9}(C)$ &$A^r_{10}(C)$ &$A^r_{11}(C)$ &$A^r_{12}(C)$\\
\hline
$r=1$ &5 (minimal)&61 (minimal)&115 (minimal)& 150(not minimal)& 221(not minimal)&185(not minimal)&44\\\hline
 $r=2$ &- &- &-&220 (minimal)&1386(not minimal)&6240(not minimal)&12460\\\hline
 $r=3$ &- & -&-&-&66 (minimal)&1740(not minimal)&18500\\\hline
 $r=4$ &- &- &-& -&- &12 (minimal)&769\\\hline
 $r=5$ & -&- &-&-&-&-&1\\
\end{tabular}
}
\end{table}

  \end{Example}


\section{Generalized Ashikhminand-Barg theorem for $s$-minimal codes}
In this section, we provide a  generalization of the  Ashikhminand-Barg theorem for $s$-minimal codes.

 \begin{Theorem}\label{T2}
Let $C$ be an  $[n,k]_q$-linear code with the $s$-generalized Hamming weight $d_{s}(C)$ and the $s$-maximum Hamming weight $D_{s}(C)$, where $1\leq s\leq k-1$ .
\begin{description}
  \item[(a)] If  $D_{s}(C)<d_{s+1}(C),$ then $C$ is a $s$-minimal code.

  \item[(b)]  If  $\frac{D_{s}(C)}{d_{s}(C)}<\frac{q^{s+1}-1}{q^{s+1} -q},$ then $C$ is a $s$-minimal code.
\end{description}
\end{Theorem}

\begin{proof}
The proofs of both Statements {\bf (a)} and {\bf (b)} are similar.
For any $U\in {\rm SUB}^{s}(C)$, we know that $|\mathrm{supp}(U)|\leq D_{s}(C) <d_{s+1}(C).$
By Theorem~\ref{T1}, we know that $U$ is a $s$-minimal subcode of $C,$ and $C$ is a $s$-minimal code.
\end{proof}

In the following example, we show that a simplex code of dimension $k$ is a $s$-minimal code  for  $1\leq s\leq k-1.$
\begin{Example}
For $k \ge  2$, let $S_{q,k}$ be a $k \times  \frac{q^k-1}{q-1}$ matrix over $\mathbb{F}_q$ such that every two columns of $S_{q,k}$ are linearly independent.
Let $\mathcal{S}_{q,k}$ be the $[\frac{q^k-1}{q-1}, k, q^{k-1}]_q$-linear code generated by $S_{q,k}.$
It is well-known that $$d_s(\mathcal{S}_{q,k})=D_s(\mathcal{S}_{q,k})=|\mathrm{supp}(U)|=\frac{q^{k}-q^{k-s}}{q-1}$$
for any $U\in {\rm SUB}^{s}(\mathcal{S}_{q,k})$ and $1\leq s\leq k.$
By  Theorem~\ref{T2}, we know that $\mathcal{S}_{q,k}$ is a $s$-minimal code for  $1\leq s\leq k-1.$
\end{Example}

In the following corollary, we prove that some  linear codes  constructed by punctured on simplex codes are $s$-minimal codes.
\begin{Corollary}
Let $\mathcal{S}_{q,k}$ be the $[\frac{q^k-1}{q-1}, k, q^{k-1}]_q$-linear code generated by $S_{q,k}.$
Let $C$ be the linear code constructed by punctured on any $t$ coordinates of $\mathcal{S}_{q,k}$.
If $1\leq s\leq k-1$ and $t<q^{k-s-1},$ then $C$ is a $s$-minimal code.
\end{Corollary}
\begin{proof}
Let $\phi_t$ be the linear map from $\mathcal{S}_{q,k}$ to $C$  by punctured on the $t$ coordinates of $\mathcal{S}_{q,k}$.
Since the minimum Hamming weight of $\mathcal{S}_{q,k}$ is $q^{k-1}$ and $t<q^{k-s-1}<q^{k-1},$
we know that $\phi_t$ is bijective.
Then there  exists $V\in {\rm SUB}^{s}(\mathcal{S}_{q,k})$ for any $U\in {\rm SUB}^{s}(C)$ such that $\phi_t(V)=U$
and $$w(V)-t\leq w(U)\leq w(V).$$
Hence we have that $$D_{s}(C) \leq D_{s}(\mathcal{S}_{q,k}) = \frac{q^k-q^{k-s}}{q-1}$$ and  $$d_{s+1}(C)\ge d_{s+1}(\mathcal{S}_{q,k})-t= \frac{q^k-q^{k-s-1}}{q-1}-t.$$
Since $$d_{s+1}(C)-D_{s}(C)\ge q^{k-s-1}-t >0,$$
we know that $C$ is a $s$-minimal code by Theorem~\ref{T2}.
\end{proof}
In the following theorem, we show that,  for any   $[n,k]_q$-linear code $C$ with a generator matrix $G$, there exists a
$s$-minimal code by adding more columns in the generator matrix $G.$
\begin{Theorem}\label{T3}
Let $C$ be an  $[n,k]_q$-linear code  with a generator matrix $G$ with the $s$-generalized Hamming weight $d_{s}(C)$ and the $s$-maximum Hamming weight $D_{s}(C)$.  For $1\leq s\leq k-1,$ there exists a integer
$$t_s =\min\{0,\, \frac{D_{s}(C)-d_{s+1}(C)}{q^{k-s-1}}\}$$
such that the $[n+t\frac{q^{k}-1}{q-1},k]_q$-linear code $\tilde{C}$ generated by the matrix
\begin{equation*}
  \tilde{G}=[G,\underbrace{S_{q,k},\cdots,S_{q,k}}_{t}]
\end{equation*}
is a  $s$-minimal code for $t\ge t_s.$

In particular, when $k\ge 2$ and $s=1,$ there exists a integer
$$t_1 =\min\{0,\, \frac{D_{1}(C)-d_{2}(C)}{q^{k-2}}\}$$
such that the $[n+t\frac{q^{k}-1}{q-1},k]_q$-linear code $\tilde{C}$ generated by the matrix
\begin{equation*}
  \tilde{G}=[G,\underbrace{S_{q,k},\cdots,S_{q,k}}_{t}]
\end{equation*}
is a  minimal code for $t\ge t_1.$
\end{Theorem}

\begin{proof}
 For any $\tilde{U}\in {\rm SUB}^{s}(\tilde{C})$,  we assume that $\tilde{U}$ is generated by $\{ \tilde{\mathbf{c}}_1,\tilde{\mathbf{c}}_2,\cdots,\tilde{\mathbf{c}}_s\}.$
 Then there exists $\mathbf{x}_i\in \mathbb{F}_q^k$  such that $\tilde{\mathbf{c}}_i=\mathbf{x}_i\tilde{G}$  for $1\leq i\leq s$.

 Let $\mathbf{c}_i=\mathbf{x}_iG$  for $1\leq i\leq s$ and $U$ be the subcode generated by $\{ \mathbf{c}_1,\mathbf{c}_2,\cdots,\mathbf{c}_s\}.$
Let $\bar{\mathbf{c}}_i=\mathbf{x}_iG$  for $1\leq i\leq s$ and $\bar{U}$ be the subcode generated by $\{ \bar{\mathbf{c}}_1,\bar{\mathbf{c}}_2,\cdots,\bar{\mathbf{c}}_s\}.$
Then we know that $$|\mathrm{supp}(\tilde{U})|=|\mathrm{supp}(U)|+t|\mathrm{supp}(\bar{U})|.$$

Let $\mathcal{S}_{q,k}$ be the simplex code generated by $S_{q,k}.$
Since $D_s(\mathcal{S}_{q,k})=d_s(\mathcal{S}_{q,k})=\frac{q^k-q^{k-s}}{q-1}$, we have that $|\mathrm{supp}(\bar{U})|=\frac{q^k-q^{k-s}}{q-1}$ for any
$\bar{U}\in  {\rm SUB}^{s}(\mathcal{S}_{q,k})$. Then  we have that
$$|\mathrm{supp}(\tilde{U})|=|\mathrm{supp}(U)|+t\frac{q^k-q^{k-s}}{q-1}.$$
Since  $d_{s+1}(\mathcal{S}_{q,k})=\frac{q^k-q^{k-s-1}}{q-1}$, we know  that
$$D_{s}(\tilde{C})=D_{s}(C)+t\frac{q^k-q^{k-s}}{q-1}$$
and $$d_{s+1}(\tilde{C})=d_{s+1}(C)+t\frac{q^k-q^{k-s-1}}{q-1}.$$
Hence $$d_{s+1}(\tilde{C})-D_{s}(\tilde{C})=d_{s+1}(C) - D_{s}(C)+tq^{k-s-1}\ge d_{s+1}(C) - D_{s}(C)+t_sq^{k-s-1}>0.$$
By  Theorem~\ref{T2}, we know that  $\tilde{C}$ is a $s$-minimal code.
\end{proof}

\begin{Example}\label{ego}
Assume that $n=12,$ $k=6$ and $q=3.$
Let $C$ be the $[12,6,6]_3$-extended Golay code with the generator matrix
$$G=\left(\begin{array}{ccccc ccccc cc}
1&0&0&0&0&0& 0&1&1&1&1&1\\
0&1&0&0&0&0& 1&0&1&2&2&1\\
0&0&1&0&0&0& 1&1&0&1&2&2\\
0&0&0&1&0&0& 1&2&1&0&1&2\\
0&0&0&0&1&0& 1&2&2&1&0&1\\
0&0&0&0&0&1& 1&1&2&2&1&0
\end{array}\right) .$$
Then we know that $C=C^{\bot}$.
Let $C_t$ be the $[12+364t,6]_q$-linear code  generated by the matrix
\begin{equation*}
  \tilde{G}=[G,\underbrace{S_{3,6},\cdots,S_{3,6}}_{t}]
\end{equation*}
By Magma, the  $s$-generalized Hamming weight  and the $s$-maximum Hamming weight  of $C$ and $C_t$ are determined in Table~\ref{t1}.
By  Theorem~\ref{T2}, we know that  $C_t$ is a $s$-minimal code for $1\leq s\leq 5$,
but $C$ is not a $s$-minimal code for $1\leq s\leq 5.$

\begin{table}[htbp]

{ \small
\caption{The parameters of  $C$ and $C_t$ and  in Example~\ref{ego}}
\label{t1}
\center
\begin{tabular}{c |c|c|c|c|c|c}
&$d_{s+1}(C)$&$D_s(C)$  &$s$-minimal&$d_{s+1}(C_t)$&$D_s(C_t)$ &$s$-minimal \\ \hline
$s=1$&8 & 12&no&8+324t&12+243t        &yes                         \\ \hline
$s=2$& 9&12& no&9+ 351t&           12+ 324t &yes                     \\ \hline
$s=3$&10 &12& no&10+ 360t&            12        + 351t &yes            \\ \hline
$s=4$& 11&12& no&11+ 363t&          12           + 360t    &yes        \\ \hline
$s=5$& 12&12& no&12+ 364t&           12      + 363t         &yes       \\
\end{tabular}

}
\end{table}

\end{Example}


\section{$s$-Minimal codes from Solomon-Stiffler codes}
The minimum Hamming distance of the dual code of  linear codes constructed in Theorem~\ref{T3} is two.
In the following theorem, we construct several infinite families of $s$-minimal codes from Solomon-Stiffler codes such that the minimum Hamming distance of the dual code of those codes is three.

For $1\leq i\leq t$ and positive integers $u_i$ with $1\leq u_i\leq k-1$, let $U_i$ be a subspace of $\mathbb{F}_q^k$ with dimension $u_i$ and let $G_i$ be a  $k \times \frac{q^{u_i}-1}{q-1}$ submatrix  of $S_{q,k}$ such that each column of $G_i$ is in $U_i$. Also let $[G_1,G_2, \cdots, G_t]$ be a submatrix of $S_{q,k}$ such that each column of $[G_1,G_2, \cdots, G_t]$ is in $\bigcup_{i=1}^t U_i$.

If
$$
\sum_{i=1}^{t} u_i\leq k,~~u_i\neq u_j~~\text{and}~~U_i\bigcap U_j=\{ {\bf 0} \} $$
for $1\leq i\neq j\leq t$, then the linear code $C$ with the generator matrix
$$
G=\big{[}S_{q,k}\backslash[G_1, G_2, \cdots,G_t]\big{]}
$$
is called the {\it projective Solomon-Stiffler code} \cite{SS}. According  to \cite{SS}, the projective Solomon-Stiffler code is a Griesmer code with parameters $$\big{[}\frac{q^k-1}{q-1}-\sum_{i=1}^{t}\frac{q^{u_i}-1}{q-1},k,
q^{k-1}-\sum_{i=1}^{t}q^{u_i-1}\big{]}_q.$$
The minimum Hamming distance of the dual code of projective Solomon-Stiffler codes
is great than 2.
\begin{Theorem}\label{T31}
Assume that $t\leq q-1$, $1\leq u_1<u_2<\cdots<u_t\leq k-s-1$ and $\sum_{i=1}^{t} u_i\leq k$. There exists a  $[\frac{q^{k}-1}{q-1} -\sum_{i=1}^{t}\frac{q^{u_i}-1}{q-1},k,q^{k-1}-\sum_{i=1}^{t}q^{u_i-1}]_q$-linear code $C$, which is a $s$-minimal code and a Griesmer code.
In particular, when $t=1$ or $ 2,$ we have follows:
\begin{description}
  \item[(a)] Assume $t=1$ and $1\leq u_1\leq k-s-1$. There exists a  $[\frac{q^{k}-1}{q-1} -\frac{q^{u_1}-1}{q-1},k,q^{k-1}-q^{u_1-1}]_q$-linear code $C$, which is a $s$-minimal code.
      Also $C$ is a Griesmer code and the weight distribution of the linear code $C$ is completely determined in the following table:
$$\begin{array}{c|c}
                                         \text{weight} & \text{multiplicity} \\
                                           \hline
                                         0 & 1 \\
                                    q^{k-1}-q^{u_1-1}& q^{k} -q^{k-u_1}\\
                                         q^{k-1} & q^{k-u_1}-1\\
                                        \end{array}  .$$

  \item[(b)] Assume $t=2$ and $1\leq u_1<u_2\leq k-s-1$ and $q\ge 3$. There exists a  $[\frac{q^{k}-1}{q-1} -\frac{q^{u_1}-1}{q-1}-\frac{q^{u_2}-1}{q-1},k,q^{k-1}-q^{u_1-1}-q^{u_2-1}]_q$-linear code $C$, which is a $s$-minimal code.
       Then $C$ is a Griesmer code and the weight distribution of the linear code $C$ is completely determined in the following table:
$$\begin{array}{c|c}
                                         \text{weight} & \text{multiplicity} \\
                                           \hline
                                         0 & 1 \\
                                    d& q^{k} -q^{k-u_1}\\
                                          d+q^{u_1-1}& q^{k-u_1}-q^{k-u_2}\\
                                          d+q^{u_1-1}+q^{u_2-1}& q^{k-u_2}-1\\
                                        \end{array}  .$$
\end{description}
  \end{Theorem}

\begin{proof}
 Let $\mathbf{e}_i\in \mathbb{F}_q^k$ be the vector with all $0$s except for a $1$ in the $i$th coordinate.
Assume $u_0=0$ and $U_i$ is generated by $\{\mathbf{e}_{u_{i-1}+1},\,\mathbf{e}_{u_{i-1}+2},\cdots ,\mathbf{e}_{u_{i-1}+u_i}\}$ for $1\leq i\leq t.$
Let $G_i$ be the $k \times \frac{q^{u_i}-1}{q-1}$ submatrix of $S_{q,k}$ such that each column of $G_i$ is in $U_i$.

{\bf (a)} Let $C$ be the $[n,k]_q$-linear code with the generator matrix
\begin{equation*}
 G=\big{[}S_{q,k}\backslash[G_1, G_2, \cdots,G_t]\big{]},
\end{equation*}
where $n=\frac{q^{k}-1}{q-1}-\sum_{i=1}^{t}\frac{q^{u_i}-1}{q-1}$.
For any subspace $U\in {\rm SUB}^{s}(C)$, there exists a subspace $V\in {\rm SUB}^{s}( \mathbb{F}_q^k)$ such that $U=\{ \mathbf{y}G\,|\,  \mathbf{y}\in V\}.$
By Lemma~\ref{weight}, we know that the support weight of $U$ is
\begin{equation}\label{uu31}
 w(U)= n-m_{G}( V^{\bot}),
\end{equation}
where $\dim(V^{\bot})=k-s.$

By the definition of the function $m_{G}$, we have that
\begin{equation}\label{oor31}
  m_{S_{q,k}}(V^{\bot})=m_{G}(V^{\bot})+\sum_{i=1}^{t} m_{G_i}(V^{\bot})  .
\end{equation}

By the definition of $S_{q,k}$,
we know that $m_{S_{q,k}}(V^{\bot})=\frac{q^{k-s}-1}{q-1} .$
By Equalities~(\ref{uu31}) and (\ref{oor31}), we have that
\begin{equation}\label{oo31}
  w(U)=\frac{q^{k}-q^{k-s}}{q-1}-\sum_{i=1}^{t}\frac{q^{u_i}-1}{q-1}+ \sum_{i=1}^{t}m_{G_i}(V^{\bot}).
\end{equation}
Since $u_i\leq k-s-1,$ we have that $$\max\{0, u_i-s-1\}\leq \dim(U_i \bigcap V^{\bot})\leq u_i$$  and
$$\max\{0,  \frac{q^{u_i-s-1}-1}{q-1}\}\leq m_{G_i}(V^{\bot})\leq  \frac{q^{u_i}-1}{q-1} .$$
Then we have that  $$d_{s+1}(C)\ge \frac{q^{k}-q^{k-s-1}}{q-1}-\sum_{i=1}^{t}\frac{q^{u_i}-1}{q-1}$$
and  $$D_{s}(C)\leq \frac{q^{k}-q^{k-s}}{q-1}.$$
Since  $t\leq q-1$ and $u_i\leq k-s-1,$
we have that $  (q-1)q^{k-s-1}\ge  tq^{k-s-1}\ge \sum_{i=1}^{t}q^{u_i} $
and  $$d_{s+1}(C)-D_{s}(C)\ge \frac{q^{k-s}-q^{k-s-1}}{q-1}-\sum_{i=1}^{t}\frac{q^{u_i}-1}{q-1}=\frac{(q-1)q^{k-s-1}-\sum_{i=1}^{t}q^{u_i}+t}{q-1}>0,$$

By  Theorem~\ref{T2}, we know that  $\tilde{C}$ is a $s$-minimal code.

 When $t=1$ and $1\leq u_1\leq k-s-1$, the weight distribution of the linear code $C$ is completely determined in Theorem 1 of \cite{PC}.
When $t=2$ and $1\leq u_1<u_2 \leq k-s-1$, the weight distribution of the linear code $C$ is completely determined in Theorem 5 of \cite{SL}.
\end{proof}

\begin{Remark}
When $t=1$ and $1\leq u_1\leq k-s-1$, the $s$-subcode support weight distributions of the linear code $C$ in Theorem~\ref{T31} is completely determined in Theorem 1 of \cite{PC}.
When $t=2$ and $1\leq u_1<u_2 \leq k-s-1$,  the $s$-subcode support weight distributions of the linear code $C$ Theorem~\ref{T31} is completely determined in Theorem 5 of \cite{SL}.
\end{Remark}

\begin{Example}\label{ee12pa}
Assume $q=u_1=2,$ $u_2=3$ and $k=5$.
Let
$$ G_1=\left(\begin{array}{ccc }
      1&0&1\\
      0&1&1\\
      0&0&0\\
      0&0&0\\
      0&0&0
\end{array}\right) \text{ and } G_2=\left(\begin{array}{cc ccccc}
      0&0&0&0&0 &0&0\\
      0&0&0&0&0 &0&0\\
      1&0&1&0&1&0&1\\
      0&1&1&0&0&1&1 \\
      0&0&0&1&1&1&1
\end{array}\right) .$$
Let $C_1$ be the $[28,5,14]_2$-linear code with the generator matrix $[S_{2,5}\backslash G_1]$,
and let $C_2$ be the $[24,5,11]_2$-linear code with the generator matrix $[S_{2,5}\backslash  G_2].$
The linear codes $C_1$ and $C_2$ are Griesmer codes.
By Magma, the  generalized Hamming weights
and the maximum Hamming weights of $C_1$ and $C_2$ are listed in Table~\ref{tra}.
\begin{table}[htbp]

{ \small
\caption{Parameters of  $C_1$ and $C_2$ in Example~\ref{ee12pa}}
\label{tra}
\center

\begin{tabular}{c |c|c|c|c|c|c}
&$d_{s+1}(C_1)$&$D_s(C_1)$  &$s$-minimal&$d_{s+1}(C_2)$&$D_s(C_2)$ &$s$-minimal \\ \hline
$s=1$&21 & 16&yes  &17&16       &yes                         \\ \hline
$s=2$& 25&24& yes  &21& 23&no                     \\ \hline
$s=3$&27 &28& no&23&           24 &no        \\ \hline
$s=4$& 28&28& no&24&          24    &no      \\ \hline
\end{tabular}

}
\end{table}
\end{Example}

\begin{Example}\label{ee12p}
Assume $u_1+1=u_2=2,$ $q=3$ and $k=5$.
Let
$$ G_1=\left(\begin{array}{c }
      1\\
      0\\
      0\\
      0\\
      0
\end{array}\right) \text{ and } G_2=\left(\begin{array}{cc cc}
      0&0&0&0\\
      0&1&1&1 \\
      1&0&1&2\\
       0&0&0&0\\
      0&0&0&0
\end{array}\right) .$$
Let $C_1$ be the $[117,5,78]_3$-linear code with the generator matrix $[S_{3,5}\backslash G_2]$,
and let $C_2$ be the $[116,5,77]_3$-linear code with the generator matrix $\big{[}S_{3,5}\backslash   [G_1,  G_2]\big{]}.$
The linear codes $C_1$ and $C_2$ are Griesmer codes.
By Magma, the  generalized Hamming weights
and the maximum Hamming weights of $C_1$ and $C_2$ are listed in Table~\ref{tr}.
\begin{table}[htbp]

{ \small
\caption{Parameters of  $C_1$ and $C_2$ in Example~\ref{ee12p}}
\label{tr}
\center

\begin{tabular}{c |c|c|c|c|c|c}
&$d_{s+1}(C_1)$&$D_s(C_1)$  &$s$-minimal&$d_{s+1}(C_2)$&$D_s(C_2)$ &$s$-minimal \\ \hline
$s=1$&104 & 81&yes  &103&81       &yes                         \\ \hline
$s=2$& 113&108& yes  &112& 108&yes                   \\ \hline
$s=3$&116 &117& no& 115&  116          &no        \\ \hline
$s=4$& 117&117& no& 116&   116           &no      \\ \hline
\end{tabular}

}
\end{table}

\end{Example}

\section{$s$-Minimal codes violating the generalized Ashkhmin-Barg condition}

In this section, we give constructions of $s$-minimal codes that violate the generalized Ashikhmin-Barg condition, and generalizing the constructions for minimal codes violating the Ashikhmin-Barg condition in \cite{Chen}.

\begin{Theorem}\label{t-7.1}
    Let $C$ be an $[n,k]_q$-linear  code with the $s$-generalized Hamming weight $d_{s}(C)$ and the $s$-maximum Hamming weight $D_{s}(C)$.
    Suppose $C$ satisfies the conditions of Theorem \ref{T2},
    $$\frac{d_{s}(C)}{D_{s}(C)} >\frac{q^{s+1}-q}{q^{s+1}-1},$$
    thus $C$ is a $s$-minimal code.
    Set $$n'=\left\lceil \frac{(q^{s+1}-1)d_{s}(C)}{q^{s+1}-q} \right\rceil-D_{1}(C).$$
    Then we construct an explicit $q$-ary $s$-minimal code $C' \subset \mathbb{F}_q^{n+n'}$ with the $s$-maximum Hamming weight at least $D_{1}(C)+n'$ and the $s$-generalized Hamming weight $d_{s}(C)$ satisfying
    $$\frac{d_{s}(C')}{D_{s}(C')} \leq \frac{q^{s+1}-q}{q^{s+1}-1}.$$
\end{Theorem}

\begin{proof}
    Let $G$ be a generator matrix of $C$, where the first row ${\bf r}_1$ is the codeword $(c_1,\ldots, c_n)$ with maximum weight $D_{1}(C)$, and the second row ${\bf r}_2$ is a codeword with the minimum weight $d_{1}(C)$.
    Let $G'$ be a $k \times (n+n')$ matrix with the first row ${\bf r}_1'=(c_1',\ldots, c_n',c_1, \ldots, c_n)$, where $c_1', \ldots, c_n'$ are nonzero elements in $\mathbb{F}_q$, the remaining rows are of the form ${\bf r}_i'=({\bf 0}, {\bf r}_i)$, where ${\bf 0}$ denotes the zero vector in $\mathbb{F}_q^{n'}$ and ${\bf r}_i$ is the $i$-th row of $G$, for $i=2,\ldots,k$.

    It is clear that the rank of the matrix $G'$ remains $k$.
    Denote by $C'$ the linear code in $\mathbb{F}_q^{n+n'}$ generated by $G'$.
    According to the defining of $C'$, the $s$-generalized Hamming weight of $C'$ remains $d_{s}(C)$.
    The $s$-maximum Hamming weight of $C'$ is at least $n'+D_{1}(C)$, since the Hamming weight of the first row of $G'$ is $n'+D_{1}(C)$.
    We now prove that $C'$ is a minimal code.
    Let $U_1, U_2\in {\rm SUB}^{s}(C)$ satisfying
    $$\mathrm{supp}(U_2) \subseteq \mathrm{supp}(U_1).$$
    We analysis the following two cases.

    {\em Case 1}. $U_1$ is the subspace of $C'$ formed by some linear combinations of rows ${\bf r}_2',\ldots, {\bf r}_k'$ of $G'$.
    In this case, the support of $U_1$ is in the last $n$ coordinate positions.
    Then the support of $U_2$ is also in the last $n$ coordinate positions, and $U_2$ is some linear combinations of ${\bf r}_2',\ldots, {\bf r}_k'$ rows of $G'$.
    Hence, from the minimality of the linear code $C$, we have $U_2=U_1$.

    {\em Case 2}. $U_1$ is the subspace of $C'$ formed by some linear combinations of rows ${\bf r}_1', \ldots, {\bf r}_k'$ of $G'$.
    Suppose that
    $${\bf x}_1=\mu_1{\bf r}_1'+\mu_2 {\bf c}\in U_1,$$
    where $\mu_1\neq 0$, and ${\bf c}$ is a linear combination of rows ${\bf r}_2',\ldots, {\bf r}_k'$ in $G'$.
    Let $U_i'$ be the subspace in $\mathbb{F}_q^n$ obtained by puncturing the first $n'$ coordinates of all vectors in $U_i$, $i=1,2$.
    They are two subspaces of code $C$, and we have $\mathrm{supp}(U_2') \subseteq  \mathrm{supp}(U_1')$.
    Then $U_2'=U_1'$ according to the minimality of $C$.
    Therefore, $U_2$ cannot be a subspace formed by some linear combinations of only ${\bf r}_2', \ldots, {\bf r}_k'$ in the generator matrix $G'$.
    Then the support of $U_2$ has to include the first $n'$ coordinate positions.
    From a similar argument as above, we also conclude that $U_2=U_1$.
    Then the linear code $C'$ is $s$-minimal.

    The $s$-maximum Hamming weight of $C'$ is at least
    $$D_s(C') \geq D_1(C)+n'\ = \left\lceil \frac{(q^{s+1}-1)d_{s}(C)}{q^{s+1}-q} \right\rceil.$$
    It is easy to verify that
    $$\frac{d_s(C')}{D_s(C')} \leq \frac{d_s(C)}{D_1(C)+n'} \leq \frac{q^{s+1}-q}{q^{s+1}-1}.$$
    The $s$-minimal code $C'$ violates the generalized Ashkhmin-Barg condition.
    The conclusion is proved.
\end{proof}

Next, we provide two simple examples. The minimality of these codes has also been verified by Magma.

\begin{Example}\label{e-1-xie}
    Let $C$ be the $2$-minimal linear $[28,5,14]_2$ code with $D_1(C)=16$ and $d_2(C)=21$ constructed in Example \ref{ee12pa}.
    Let  $n'=\left\lceil \frac{(2^{3}-1)\times21}{2^{3}-2} \right\rceil-16=9$. Then we can construct a $2$-minimal code $C'$ violating the generalized Ashkhmin-Barg condition, with parameters $[37, 5, 14]_2$, $d_2(C') = 21$ and $D_2(C')=33$, since $\frac{d_2(C')}{D_2(C')}=\frac{21}{33}<\frac{6}{7}$.
\end{Example}

\begin{Example}\label{e-2-xie}
    Let $q=2$, $k=5$, and
    $$ G_1=\left(\begin{array}{ccc }
      1&0&1\\
      0&1&1\\
      0&0&0\\
      0&0&0\\
      0&0&0
\end{array}\right), G_2=\left(\begin{array}{cc ccccc}
      0\\
      0\\
      1\\
      0\\
      0
\end{array}\right), \text{ and } G_3=\left(\begin{array}{cc ccccc}
      0\\
      0\\
      0\\
      1\\
      0
\end{array}\right)$$
    Let $C$ be the linear $[26,5,12]_2$ code with generator matrix $[S_{2,5}\backslash G_1\backslash G_2 \backslash G_3]$.
    Using magma, we obtain that $C$ is a $2$-minimal code with $D_1(C)=16$ $d_2(C)=19$, and $D_2(C)=23$.
    Let $n'=\left\lceil \frac{(2^{3}-1)\times 19}{2^{3}-2} \right\rceil-16=7$. Then we can construct a $2$-minimal code $C'$ violating the generalized Ashkhmin-Barg condition, with parameters $[33, 5, 12]_2$, $d_2(C') = 19$ and $D_2(C')=30$, since $\frac{d_2(C')}{D_2(C')}=\frac{19}{30}<\frac{6}{7}$.
\end{Example}

\section*{Acknowledgement}
In this work, 
Hao Chen was supported by National Natural Science Foundation of China (Grant No. 62032009).
Xu Pan was supported by National Natural Science Foundation of China (Grant No. 12401689).


\end{document}